\documentclass{article}

\usepackage{arxiv}

\usepackage{graphicx}
\usepackage{cite}
\usepackage{amsmath,amssymb,amsfonts}
\usepackage{hyperref}
\usepackage{mathrsfs}

\usepackage{amsmath}
\newtheorem{theorem}{Theorem}
\newtheorem{definition}{Definition}
\newtheorem{proof}{Proof}

\newsavebox{\savepar}
\newenvironment{boxit}{\begin{center} \begin{lrbox}{\savepar}
			\begin{minipage}[b]{140mm}}
			{\end{minipage}\end{lrbox}
		\fbox{\usebox{\savepar}
} \end{center}}

\def\Z{{\sf Z}}

\title{Zero knowledge Proofs for Cloud Storage Integrity Checking}

\author{Faen Zhang, Xinyu Fan, Pengcheng Zhou, Wenfeng Zhou\\
	\{faenzhang, pengchengzhou\}ainnovation@gmail.com\\
	fanxinyu@ainnovation.com, zhouwenfeng@interns.ainnovation.com\\
	AInnovation Technology Ltd.} 

\begin{document}
\maketitle

\begin{abstract}
	With the wide application of cloud storage, cloud security has become a crucial concern. Related works have addressed security issues such as data confidentiality and integrity, which ensure that the remotely stored data are well maintained by the cloud. However, how to define zero-knowledge proof algorithms for stored data integrity check has not been formally defined and investigated. We believe that it is important that the cloud server is unable to reveal any useful information about the stored data. In this paper, we introduce a novel definition of data privacy for integrity checks, which describes very high security of a zero-knowledge proof. We found that all other existing remote integrity proofs do not capture this feature. We provide a comprehensive study of data privacy and an integrity check algorithm that captures data integrity, confidentiality, privacy, and soundness.
	\keywords{Data Integrity \and Data Privacy \and Cloud Storage \and Cyber Security \and Cryptography.}
\end{abstract}
\section{Introduction}
Cloud computing offers different types of computational services to end users via computer networks, demonstrating a huge number of advantages. It has been becoming a trend that
individuals and IT enterprises store data remotely on the cloud in a flexible on-demand manner, which has become a popular way
of data outsourcing. This can greatly reduces the burden of storage management and maintenance and brings a great advantage of universal data access and convenience to users. In fact, cloud storage has become one of the main parts in cloud computing where user data are stored and maintained by cloud servers. It allows users to access their data via computer networks at anytime and from anywhere.

Despite the great benefits provided by cloud computing, data
security is a very important but challenging problem that must be solved. One of the major concerns of data security is data integrity in a remote storage system \cite{CSA,Arrington06}. Although storing data on the cloud is attractive, it does not always offer any guarantee on data integrity and retrievability. Simple data integrity check in a remote data storage can be done by periodically examining the data files stored on the cloud server, but such an approach can be very expensive if the amount of data is huge. An interesting problem is to check data integrity remotely without the need of accessing the full copy of data stored on the cloud server.
For example, the data owner possesses some verification token (e.g. a digest of the data file \cite{DeswarteQS04,FilhoB06}), which is very small compared with the stored dataset. However, a number of security issues have been found in previous research \cite{WangWLRL09,WangWRL10,WangWRLL11}. Several techniques, such as Proof of Retrievability (POR) \cite{ShachamW08,JuelsK07} and Third Party Auditing (TPA) \cite{WangRLL10,WangWRLL11,WangChowWangRenLou}, have been proposed to solve the above data integrity checking
problem with public auditability. POR is loosely speaking a kind of Proof of Knowledge (POK) \cite{BellareG92} where the knowledge is the data file, while TPA allows any third party (or auditor) to perform the data integrity checking on behalf of the data owner just based on some public information (e.g. the data owner's public key). Several schemes with public auditability have been proposed in the context of ensuring remotely stored data integrity under different system and security models \cite{JuelsK07,ShachamW08,AtenieseBCHKKPS11,WangWRLL11}.

Intuitively, it is important that an auditing process should not introduce new vulnerabilities of unauthorized information leakage towards the data security \cite{ShahBMS07}. The previous efforts in Remote Integrity Checking (DIC) accommodate several security features including data integrity and confidentiality, which mainly ensure secure maintenance of data. However, they do not cover the issue of data {\em privacy}, which means that the communication flows
(DIC proofs) from the cloud server should not reveal any useful
information to the adversary. Intuitively, by ``privacy'', we mean that an adversary should not be able to distinguish which file has been uploaded by the client to the cloud server. We refer it as {\em Zero Knowledge}. We believe that it is very
important to consider such privacy issues adequately in protocol designs. Taking some existing TPA based DIC proofs
\cite{WangWRL10,WangChowWangRenLou,WangWRLL11} as an example, the proof sent by the cloud server to the auditor does not allow the auditor to recover the file, but the auditor can still distinguish which file (among a set of possible files) is involved in the DIC proof, which is clearly undesirable.

In this paper, we propose an Zero Knowledge-based definition
of data privacy (DIC-Privacy) for TPA based DIC
protocols. We show that two recently published DIC schemes
\cite{WangChowWangRenLou,WangWRLL11} are insecure under our new
definition, which means some information about the user file is
leaked in the DIC proof. We then provide an new construction to
demonstrate how DIC-privacy can be achieved. We show that by
applying the Witness Zero Knowledge proof technique
\cite{GrothS08}, we are able to achieve DIC-privacy in DIC
protocols. To the best of our knowledge, our construction is the
first scheme that can achieve DIC-privacy.

\noindent\textbf{Paper Organization.} The rest of the paper is
organized as follows. In Section 2, we describe the security model
and definition of data privacy for DIC proofs. In Section 3, we
analyze the DIC protocols by Wang {\em et al.} and show why their DIC protocols fail to capture data privacy. In Section 4, we demonstrate how data privacy can be achieved with a witness
Zero Knowledge proof. We also provide the definition of
soundness for DIC proofs and show the soundness of our protocol
based on witness Zero Knowledge proof. We conclude the paper
in Section 6.

\section{Definitions and Security Model}

\textbf{DIC Protocols.} We will focus on TPA based Data Integrity Checking (DIC) protocols for cloud data storage systems. The protocol involves three entities: the cloud storage server, the cloud user, and the third party auditor (TPA). The cloud user relies on the cloud storage server to store and maintains his/her data. Since the user no longer keeps the data locally, it is of critical importance for the user to ensure that the data are correctly stored and maintained by the cloud server. In order to avoid periodically data integrity verification, the user will resort to a TPA for checking the integrity of his/her outsourced data. To be precise, an DIC protocol for cloud storage consists of five algorithms:

\begin{itemize}
	\item {\sf KeyGen}:  Taking as input a security parameter $\lambda$, the algorithm {\sf KeyGen} generates the public and private key pair $(pk,sk)$ of a cloud user (or data owner).
	
	\item {\sf TokenGen}: Taking as input a file $F$ and the user
	private key $sk$, this algorithm generates a file tag $t$ (which includes a file name $name$) and an authenticator $\sigma$ for $F$. The file and file tag, as well as the authenticator are then stored in the cloud server.
	
	%
	
	\item {\sf Challenge}: Given the user public key $pk$ and a file tag
	$t$, this algorithm is run by the auditor to generate a random
	challenge $chal$ for the cloud server.
	
	
	\item {\sf Respond}: Taking as input $(F, t, \sigma, chal)$,
	this algorithm outputs a proof ${\cal P}$, which is used to prove the integrity of the file.

	\item {\sf Verify}: Taking as input $(pk, t, chal, {\cal P})$, the
	algorithm outputs either True or False.
	
\end{itemize}

\noindent\textbf{DIC Privacy.} We define the data privacy for DIC
proofs via an {\em Zero Knowledge} game between a simulator
${\cal S}$ (i.e. the cloud server or prover) and an adversary
${\cal A}$ (i.e. the auditor or verifier).

\medskip

\noindent{\bf Setup}: The simulator runs {\sf KeyGen} to generate
$(sk, pk)$ and passes $pk$ to the adversary $\cal A$.

\noindent {\bf Phase 1}: ${\cal A}$ is allowed to make Token
Generation queries. To make such a query, $\cal A$ selects a file
$F$ and sends it to ${\cal S}$. ${\cal S}$ generates a file tag
$t$, an authenticator $\sigma$, and then returns $(t,\sigma)$ to
${\cal A}$.

\noindent {\bf Phase 2}: ${\cal A}$ chooses two different files
$F_0, F_1$ that have not appeared in Phase 1, and send them to
${\cal S}$. $\cal S$ calculates $(t_0,\sigma_0)$ and
$(t_1,\sigma_1)$ by running the {\sf TokenGen} algorithm. $\cal S$
then tosses a coin $b \in \{0,1\}$, and sends $t_b$ back to $\cal
A$. ${\cal A}$ generates a challenge $chal$ and sends it to ${\cal
	S}$. ${\cal S}$ generates a proof ${\cal P}$ based on
$(F_b,t_b,\sigma_b)$ and ${\cal A}$'s challenge $chal$ and then
sends ${\cal P}$ to ${\cal A}$. Finally, ${\cal A}$ outputs a bit
$b'$ as the guess of $b$. The process is illustrated in Figure
\ref{fig:ind}.

Define the advantage of the adversary ${\cal A}$ as
\[Adv_{\cal A}(\lambda) = |\Pr[b'=b]-1/2|.\]

\begin{definition} An DIC proof has \textit{Zero Knowledge} if for
	any polynomial-time algorithm, $Adv_{\cal A}(\lambda)$ is
	a negligible function of the security parameter $\lambda$.
\end{definition}

\begin{figure*}[ht]
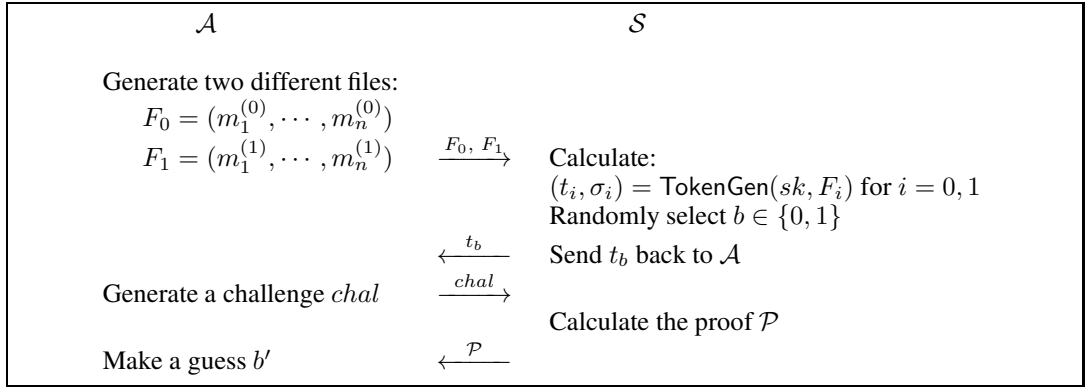

	\begin{boxit}
		\begin{center}
			\begin{tabular}{l@{\hspace{0.5cm}}c@{\hspace{0.5cm}}l}
				{ ~~~~~~~~~~~~~$ \cal A $}            &       &   { ~~~~~~~~~~~$\cal S$ }\\
				&   &\\
				Generate two different files: &&\\
				~~~~~~$F_0 = (m_1^{(0)},\cdots,m_n^{(0)})$          &       &  \\
				~~~~~~$F_1 = (m_1^{(1)},\cdots,m_n^{(1)})$     &  $\stackrel{F_0,~
					F_1}{\overrightarrow{\hspace{1cm}}}$ & Calculate: \\
				&& $(t_i,\sigma_i) =
				{\sf TokenGen}(sk,F_i)$ for $i = 0,1$\\
				& & Randomly select $b \in \{0,1\}$   \\
				& $\stackrel{t_b}{\overleftarrow{\hspace{1cm}}}$  & Send $t_b$ back to $\cal A$ \\
				Generate a challenge $chal$  &  $\stackrel{chal}{\overrightarrow{\hspace{1cm}}}$          & \\
				&& Calculate the proof $\cal P$\\
				Make a guess $b'$   & $\stackrel{\cal
					P}{\overleftarrow{\hspace{1cm}}}$ &
			\end{tabular}
		\end{center}
	\end{boxit}
	\caption{Zero knowledge Proof game run between $\cal A$ and $\cal S$}
	\label{fig:ind}
\end{figure*}

\section{Privacy Analysis of Existing DIC Protocols}

\subsection {Notations and Preliminaries}

Before describing some existing DIC protocols, we first introduce
some notations and tools used in those protocols. We denote $F$
the data file to be stored in the cloud. It is decomposed as a
sequence of $n$ blocks $m_1, ..., m_n \in \Z_p $ for some large
prime $p$. We denote by $H(\cdot)$ and $h(\cdot)$ cryptographic
hash functions.

Let $G_1, G_2$ and $G_T$ be multiplicative cyclic groups of prime
order $p$. Let $g_1$ and $g$ be generators of $G_1$ and $G_2$,
respectively. A bilinear map is a map $e: G_1 \times G_2 \rightarrow
G_T$ such that for all $u\in G_1$, $v \in G_2$ and $a,b \in  \Z_p$,
$e(u^a,v^b) = e(u,v)^{ab}$. Also, the map $e$ must be efficiently
computable and non-degenerate (i.e. $e(g_1, g) \ne 1$). In addition,
let $\psi$ denote an efficiently computable isomorphism from $G_2$
to $G_1$, with $\psi(g) = g_1$ \cite{BonehLS04}.

\subsection{A DIC Protocol by Wang {\em et al.} \cite{WangWRLL11}}

In \cite{WangWRLL11}, Wang {\em et al.} presented a DIC protocol based on
Merkle Hash Tree (MHT) \cite{Merkle80}. Their protocol works as
follows.
\medskip

\noindent{\bf Setup Phase}: The cloud user generates the keys and
authentication tokens for the files as follows.
\medskip

\noindent{\sf KeyGen}: The cloud user runs {\sf KeyGen} to generate the
public and private key pair. Specifically, the user generates a
random verification and signing key pair $(spk,ssk)$ of a digital
signature scheme, and set the public key $pk=(v,spk)$ and
$sk=(x,ssk)$ where $x$ is randomly chosen from $\Z_p$ and $v= g^x$.
\medskip

\noindent{\sf TokenGen}:  Given a file $F = (m_1,m_2,\cdots,m_n)$,
the client chooses a file name $name$, a random element $u \in G_1$
and calculates the file tag
\[ t=name \| n \| u \| SSig_{ssk}(name  \| n \| u),
\]
and authenticators ${\sigma_i}  = {(H(m_i)  \cdot {u^{m_i}})}^{x}$
where $H$ is a cryptographic hash function modeled as a random
oracle. The client then generates a root $R$ based on the
construction of Merkle Hash Tree (MHT) where the leave nodes of the
tree are an ordered set of hash values $H(m_i) (i = 1,2,\cdots,n)$.
The client then signs the root $R$ under the private key $x$:
$sig_{sk}(H(R))= (H(R))^{x}$ and  sends $\{F, t,
\{\sigma_i\}, sig_{sk}(H(R))\}$ to the cloud server.
\medskip

\noindent{\bf Audit Phase}: The TPA first retrieves the file tag $t$
and verifies the signature $SSig_{ssk}(name\|n\|u)$ by using $spk$.
The TPA then obtains $name$ and $u$.
\medskip

\noindent{\sf Challenge}: To generate $chal$, TPA picks a random
subset $I=\{s_1, s_2, s_3,...,s_c\}$ of set $[1,n]$, where $s_1 \leq
\cdots \leq s_c$. Then,  the TPA sends a challenge $chal = {\{i,
	\nu_i\}}_{i \in I}$ to the cloud server where $\nu_i$ is randomly
selected from $\Z_p$.
\medskip

\noindent{\sf Response}: Upon receiving the challenge $chal = {\{i,
	\nu_i\}}_{i \in I}$, the cloud server computes $\mu = \sum_{i \in
	I}\nu_im_i$ and $\sigma = \prod_{i \in I}\sigma_i^{\nu_i}$. The
cloud server will also provide the verifier with a small amount of
auxiliary information $\{\Omega_i\}_{i \in I}$, which are the node
siblings on the path from the leaves ${H(m_i)}_{i \in I}$ to the
root $R$ of the MHT. The server sends the proof ${\cal P}=\{\mu,
\sigma, \{H(m_i),\Omega_i\}_{i \in I}, sig_{sk}(H(R))\}$ to the TPA.
\medskip

\noindent{\sf Verify}: Upon receiving the responses form the cloud
server, the TPA generates the root $R$ using $\{H(m_i),\Omega_i\}_{i
	\in I}$, and authenticates it by checking
\[
e(sig_{sk}(H(R)),g)=e(H(R),v).
\]
If the authentication fails, the
verifier rejects by emitting FALSE. Otherwise, the verifier checks
\[e(\sigma,g)=e((\prod\limits_{i=s_1}^{s_c}H(m_i)^{\nu_i})u^\mu,v).\]
If the equation holds, output True; otherwise, output False.

\subsubsection{Zero knowledge Proof Analysis}

It is easy to see that the above DIC protocol does not provide
DIC-Privacy. Let $\cal A$ denote an Zero Knowledge Proof adversary which works as
follows (also see Fig. \ref{fig:wang11}).

\begin{itemize}
	\item  $\cal A$ chooses distinct files $F_0 = (m_1^{(0)},\cdots,m_n^{(0)})$
	and $F_1 = (m_1^{(1)},\cdots,m_n^{(1)})$ where
	$m_i^{(0)} \ne m_i^{(1)}$.
	\item $\cal S$ chooses at random a file $F_b$ for $b \in \{0,1\}$ and then computes
	$t_b, \{\sigma^{(b)}_i\},$ $sig_{sk}(H(R^{(b)}))$.
	\item $\cal A$ chooses a random challenge $chal = \{i,\nu_i\}_{i \in I}$.
	\item $\cal S$ computes and sends to $\cal A$ the response
	\[
	{\cal P} = (\mu^{(b)},
	\sigma^{(b)}, \{H(m^{(b)}_i),\Omega^{(b)}_i\}_{i \in I},
	sig_{sk}(H(R^{(b)}))).
	\]
	\item $\cal A$ chooses $i \in I$ and calculates $H(m^{(0)}_i)$ and
	compare it with the received $H(m^{(b)}_i)$. If they are equal,
	output 0; otherwise, output 1.
\end{itemize}

\noindent\textbf{Probability Analysis.} It is easy to see that $\cal A$
has an overwhelming probability to guess the value of $b$ correctly
since the probability that
\[
m^{(0)}_i \ne m^{(1)}_i \wedge H(m^{(0)}_i) = H(m^{(1)}_i)
\]
is negligible since the hash function
is assumed to be a random oracle in \cite{WangWRLL11}.

\begin{figure*}[thb]
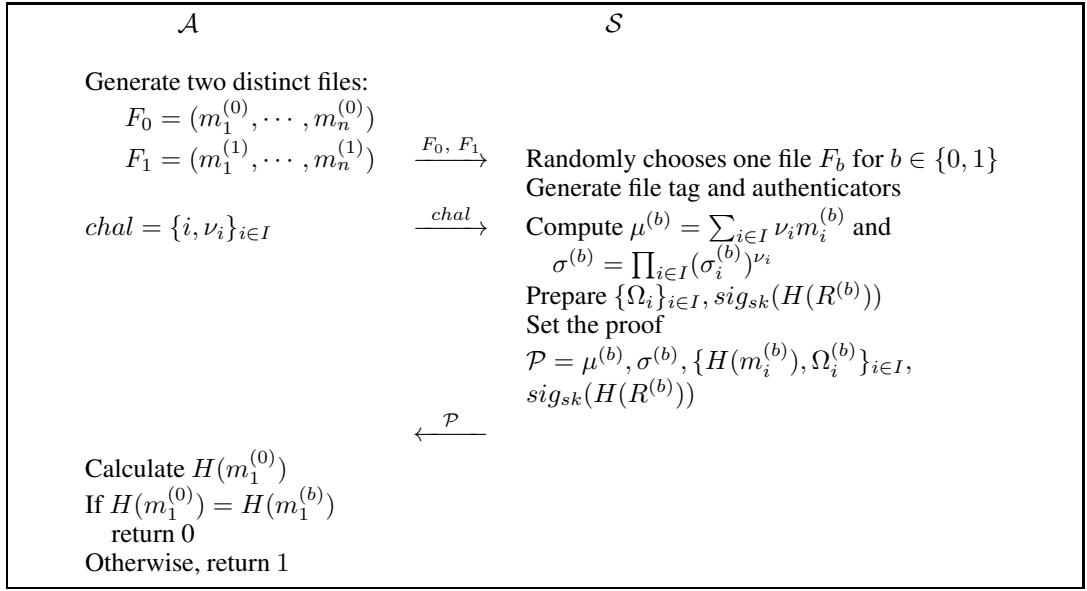

	\begin{boxit}
		\begin{center}
			\begin{tabular}{l@{\hspace{0.5cm}}c@{\hspace{0.5cm}}l}
				{ ~~~~~~~~~~~~~$\cal A$}            &       &   { ~~~~~~~~~~~$\cal S$ }\\
				&   &\\
				Generate two distinct files: &&\\
				~~~~~~$F_0 = (m_1^{(0)},\cdots,m_n^{(0)})$          &       &  \\
				~~~~~~$F_1 = (m_1^{(1)},\cdots,m_n^{(1)})$     &  $\stackrel{F_0,~ F_1}{\overrightarrow{\hspace{1cm}}}$ &  Randomly chooses one file $F_b$ for $b \in \{0,1\}$     \\
				& &  Generate file tag and authenticators \\
				$chal = \{i, \nu_i\}_{i\in I}$  &  $\stackrel{chal}{\overrightarrow{\hspace{1cm}}}$          &   Compute $\mu^{(b)} = \sum_{i\in I} \nu_i m^{(b)}_{i}$ and \\
				&& ~~~~$\sigma^{(b)} = \prod_{i\in I}(\sigma^{(b)}_{i})^{\nu_i}$  \\
				&& Prepare $\{\Omega_i\}_{i \in I}, sig_{sk}(H(R^{(b)}))$\\
				&& Set the proof\\
				&& ${\cal P} = \mu^{(b)},
				\sigma^{(b)}, \{H(m^{(b)}_i),\Omega^{(b)}_i\}_{i \in I},$\\
				&& $sig_{sk}(H(R^{(b)}))$\\
				&  $\stackrel{\cal P}{\overleftarrow{\hspace{1cm}}}$ &\\
				Calculate $H(m^{(0)}_1)$  \\
				If $H(m^{(0)}_1) = H(m^{(b)}_1)$\\
				~~~~return 0  \\
				Otherwise, return $1$
			\end{tabular}
		\end{center}
	\end{boxit}
	\caption{Zero Knowledge analysis on Wang {\em et al.}'s DIC
		Protocol \cite{WangWRLL11}.} \label{fig:wang11}
\end{figure*}

\subsection{Another Privacy Preserving DIC Protocol by Wang {\em et al.} \cite{WangChowWangRenLou}}\label{sec:wang}



In \cite{WangChowWangRenLou}, Wang {\em et al.} introduced a new DIC
protocol. Compared with the DIC protocol presented above, this new
protocol aims to achieve the additional property of privacy
preserving (i.e. the TPA cannot learn the content of the file in the
auditing process).

%

\begin{figure*}[th]
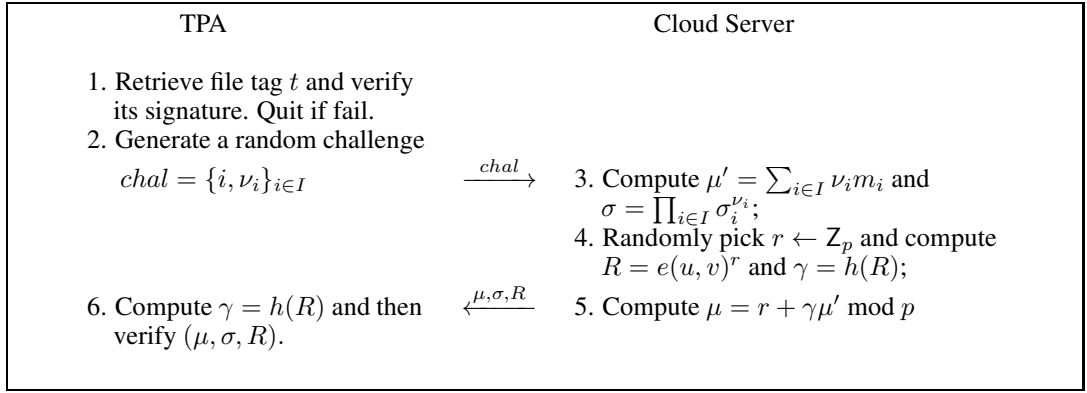

	\begin{boxit}
		\begin{center}
			\begin{tabular}{l@{\hspace{0.5cm}}c@{\hspace{0.5cm}}l}
				{ ~~~~~~~~~~~~~TPA}            &       &   { ~~~~~~~~~~~Cloud Server }\\
				&   &\\
				1. Retrieve file tag $t$ and verify            &       &  \\
				~~~~its signature. Quit if fail.   &   &  \\
				2. Generate a random challenge && \\
				~~~~ $chal = \{i, \nu_i\}_{i\in I}$  &  $\stackrel{chal}{\overrightarrow{\hspace{1cm}}}$          & 3. Compute $\mu' = \sum_{i\in I} \nu_i m_i$ and \\
				&& ~~~~$\sigma = \prod_{i\in I}\sigma_i^{\nu_i}$;  \\
				&& 4. Randomly pick $r \leftarrow \Z_p$ and compute\\
				&& ~~~~$R=e(u,v)^r$ and $\gamma = h(R)$; \\
				6. Compute $\gamma=h(R)$ and then  &  $\stackrel{\mu, \sigma, R}{\overleftarrow{\hspace{1cm}}}$ & 5. Compute $\mu = r + \gamma\mu' \mbox{ mod } p$\\
				~~~~verify $(\mu, \sigma, R).$ &&\\
				& &\\
			\end{tabular}
		\end{center}
	\end{boxit}
	\caption{The third party auditing protocol by Wang {\em et al.}
		\cite{WangChowWangRenLou}.} \label{fig:wang}
\end{figure*}

Let $(p, G_1, G_2, G_T, e, g_1, g, H, h)$ be the system parameters
as introduced above. Wang \emph{et al.}'s privacy-preserving
public auditing scheme works as follows (also see Fig. \ref{fig:wang}):
\medskip

\noindent{\bf Setup Phase}:
\medskip

\noindent{\sf KeyGen}: The cloud user runs KeyGen to generate the
public and private key pair. Specifically, the user generates a
random verification and signing key pair $(spk,ssk)$ of a digital
signature scheme, a random $x \leftarrow \Z_p$, a random element
$u \leftarrow G_1$, and computes $v \leftarrow g^x$. The user
secret key is $sk = (x,ssk)$ and the user public key is $pk =
(spk, v, u)$.
\medskip

\noindent{\sf TokenGen}: Given a data file $F = (m_1,...,m_n)$, the
user first chooses uniformly at random from $\Z_p$ a unique
identifier $name$ for $F$. The user then computes authenticator
$\sigma_i$ for each data block $m_i$ as $\sigma_i \leftarrow (H(W_i)
\cdot u^{m_i})^x \in G_1$ where $W_i = name\|i$. Denote the set of
authenticators by $\phi = \{\sigma_i\}_{1 \leq i \leq n}$. Then the
user computes $t = name\|SSig_{ssk}(name)$ as the file tag for $F$,
where $SSig_{ssk}(name)$ is the user's signature on $name$ under the
signing key $ssk$. It was assumed that the TPA knows the number of
blocks $n$. The user then sends $F$ along with the verification
metadata $(\phi,t)$ to the cloud server and deletes them from local
storage.
\medskip

\noindent{\bf Audit Phase}: 
The TPA first retrieves the file tag $t$ and verifies the
signature $SSig_{ssk}(name)$ by using $spk$. The TPA quits by
emitting $\bot$ if the verification fails. Otherwise, the TPA
recovers $name$.
\medskip

\noindent{\sf Challenge}: The TPA generates a challenge $chal$ for
the cloud server as follows: first picks a random $c$-element
subset $I = \{s_1,...,s_c\}$ of set $[1,n]$, and then for each
element $i \in I$, chooses a random value $\nu_i \in \Z_p$. The
TPA sends $chal = \{(i,\nu_i)\}_{i\in I}$ to the cloud server.
\medskip

\noindent{\sf Response}: Upon receiving the challenge $chal$, the
server generates a response to prove the data storage correctness.
Specifically, the server chooses a random element $r \leftarrow
\Z_p$, and calculates $R = e(u, v)^r \in G_T$. Let $\mu'$ denote the
linear combination of sampled blocks specified in $chal$: $\mu' =
\sum_{i \in I} \nu_i m_i$. To blind $\mu'$ with $r$, the server
computes $\mu = r + \gamma \mu' \mbox{ mod } p$, where $\gamma =
h(R) \in \Z_p$. Meanwhile, the server also calculates an aggregated
authenticator $\sigma = \prod_{i \in I}\sigma_i^{\nu_i}$. It then
sends $(\mu, \sigma, R)$ as the response to the TPA.
\medskip

\noindent{\sf Verify}: Upon receiving the response $(\mu, \sigma,
R)$ from the cloud server, the TPA validates the response by first
computing $\gamma = h(R)$ and then checking the following
verification equation
\begin{eqnarray}
R\cdot e(\sigma^\gamma,g) \stackrel{?}{=}
e((\prod\limits_{i=s_1}^{s_c}H(W_i)^{\nu_i})^\gamma\cdot u^\mu,v).
\end{eqnarray}
The verification is successful if the equation holds.

\begin{figure*}[ht]
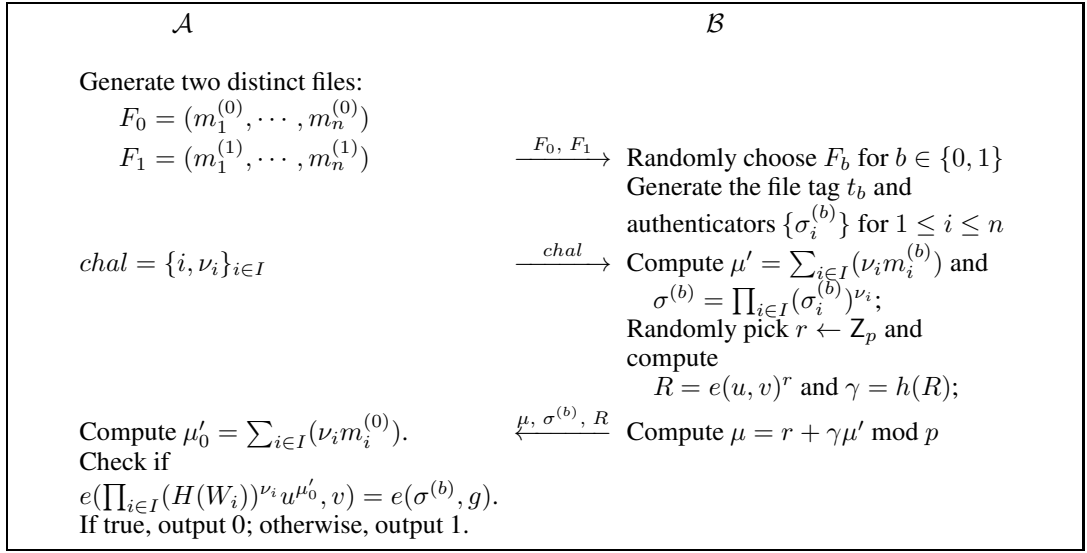

	\begin{boxit}
		\begin{center}
			\begin{tabular}{l@{\hspace{0.2cm}}c@{\hspace{0.2cm}}l}
				{ ~~~~~~~~~~~~~$\cal A$}            &       &   { ~~~~~~~~~~~$\cal B$ }\\
				&   &\\
				Generate two distinct files: &&\\
				~~~~~~$F_0 = (m_1^{(0)},\cdots,m_n^{(0)})$          &       &  \\
				~~~~~~$F_1 = (m_1^{(1)},\cdots,m_n^{(1)})$     &  $\stackrel{F_0,~ F_1}{\overrightarrow{\hspace{1.3cm}}}$ &  Randomly choose $F_b$ for $b \in \{0,1\}$     \\
				& &  Generate the file tag $t_b$ and\\
				&& authenticators $\{\sigma^{(b)}_i\}$ for $1 \le i \le n$\\
				$chal = \{i, \nu_i\}_{i\in I}$  &  $\stackrel{chal}{\overrightarrow{\hspace{1.3cm}}}$          &   Compute $\mu' = \sum_{i\in I} (\nu_i m^{(b)}_i)$ and \\
				&& ~~~~$\sigma^{(b)} = \prod_{i\in I}(\sigma^{(b)}_i)^{\nu_i}$;  \\
				&&  Randomly pick $r \leftarrow \Z_p$ and \\
				&& compute\\
				&& ~~~~$R=e(u,v)^r$ and $\gamma = h(R)$; \\
				Compute $\mu_0' = \sum_{i\in I} (\nu_i m^{(0)}_i)$. &
				$\stackrel{\mu,~ \sigma^{(b)},~ R}{\overleftarrow{\hspace{1.3cm}}}$
				&
				Compute $\mu = r + \gamma\mu' \mbox{ mod } p$\\
				Check if \\
				$e(\prod_{i\in I} (H(W_i))^{\nu_i}u^{\mu'_{0}}, v) = e(\sigma^{(b)}, g)$.\\
				If true, output 0; otherwise, output 1. \\
				
			\end{tabular}
		\end{center}
	\end{boxit}
	\caption{Zero Knowledge analysis on Wang {\em et al.} DIC
		Protocol \cite{WangChowWangRenLou}.} \label{fig:wang13}
\end{figure*}

\subsubsection{Zero Knowledge Analysis}

In \cite{WangChowWangRenLou}, it has been shown that the DIC proof
is privacy preserving. That is, the TPA cannot recover the file $F$
from the proof. This is done by concealing the value of $\mu'$.
However, we found that such a treatment could not guarantee that
there is no information leakage during the auditing process. Below
we show that Wang {\em et al.}'s scheme cannot achieve
Zero Knowledge. Let $\cal A$ denote an Zero Knowledge Proof adversary which works
as follows (also see Fig. \ref{fig:wang13}).

\begin{itemize}
	\item  $\cal A$ chooses two distinct files $F_0 =
	(m_1^{(0)},\cdots,m_n^{(0)})$ and $F_1 =
	(m_1^{(1)},\cdots,m_n^{(1)})$ such that $m_i^{(0)} \ne m_i^{(1)}$
	for $1 \le i \le n$. \item $\cal S$ randomly chooses a file $F_b$
	for $b \in \{0,1\}$ and computes the file tag $t_b$ and
	authenticators $\{\sigma^{(b)}_i\}$. \item After receiving the tag
	$t_b$, $\cal A$ chooses a random challenge $chal = \{i,\nu_i\}_{i \in
		I}$. \item $\cal S$ computes and sends to $\cal A$ the response ${\cal
		P} = (\mu,\sigma^{(b)},R)$. \item $\cal A$ computes $\mu_0' =
	\sum_{i\in I} (\nu_i m^{(0)}_i)$ and checks if
	\[
	e(\prod_{i\in I} (H(W_i))^{\nu_i}u^{\mu'_{0}}, v) = e(\sigma^{(b)}, g).
	\]
	If it is true, return 0; otherwise, return 1.
\end{itemize}

\noindent\textbf{Probability Analysis.} If $b = 0$, then
$\sigma^{(b)} = \sigma^{(0)}$ and the equation
\[
e(\prod_{i\in I}
(H(W_i))^{\nu_i}u^{\mu'_{0}}, v) = e(\sigma^{(0)}, g)
\]
always holds.
On the other hand, if $b = 1$, then $\sigma^{(b)} = \sigma^{(1)}$
and
\[
e(\prod_{i\in I} (H(W_i))^{\nu_i}u^{\mu'_{0}}, v) =
e(\sigma^{(1)}, g)
\]
holds only when
\[
\mu'_0 (= \sum_{i\in I}(\nu_i
m^{(0)}_i)) = \mu'_1 (= \sum_{i\in I} (\nu_i m^{(1)}_i)),
\]
which
happens only with probability $1/p$ for randomly selected
$\{\nu_i\}_{i \in I}$. Therefore, $\cal A$ has an overwhelming
probability to guess the value of $b$ correctly.

\section{A New DIC Protocol with DIC-Privacy}

In order to achieve the DIC-privacy, we adopt the Witness
Zero Knowledge Proof of Knowledge technique proposed by Groth and
Sahai \cite{GrothS08}. Their method can be applied to pairing
groups. Our goal is to protect both the file and the corresponding
authenticator so that the adversary cannot learn any information
about the file.

Similar to Wang {\em et al.}'s scheme \cite{WangChowWangRenLou}
reviewed in Section \ref{sec:wang}, our scheme is still based on the
``aggregate authenticator'' introduced by Shacham and Waters
\cite{ShachamW08}. That is, the cloud server will prove that the
equation
\begin{equation}
e(\sigma, g)=e((\prod\limits_{i=S_1}^{S_c} H(W_i)^ {\nu_i})
u^{\mu'},v) \label{eq1}
\end{equation}
holds, where $\mu' = \sum_{i \in I} \nu_im_i$ and $\sigma = \prod_{i
	\in I}\sigma_i^{\nu_i}$. We will treat $(u^{\mu'},\sigma)$ as the
witness when applying the Groth-Sahai proof system, and rewrite
Equation~\ref{eq1} as follows
\begin{equation}
e(\sigma, g)e(u^{\mu'},v^{-1}) = e((\prod\limits_{i=S_1}^{S_c}
H(W_i)^ {\nu_i}),v). \label{eq2}
\end{equation}

In order to protect the privacy of $\mu'$ (or $u^{\mu'}$) and
$\sigma$, the user computes an additional commitment key $\vec u =
(u_1, u_2)$ of the form
\[u_1=(u,u^\alpha),~~
u_2=(u^\tau, u^{\tau\alpha}),
\]
where $\alpha, \tau$ are selected from $\Z_p$ at random and $u$ is
the same generator of $G_1$ used in Wang {\em et al.}'s scheme.
This additional commitment key $\vec u$ is now part of the user
public key. To hide $u^{\mu'}$ and $\sigma$, the Cloud Server
computes the commitments $\vec c = (c_1,c_2)$ as
\[c_1=(c_{11},c_{12})=(u^{r_{11}+{r_{12}}\tau}, u^{\alpha(r_{11}+{r_{12}}\tau)}\sigma),
\]
\[c_2=(c_{21},c_{22})=(u^{r_{21}+{r_{22}}\tau}, u^{\alpha(r_{21}+{r_{22}}\tau)}u^{\mu'}).
\]
where $r_{i,j}$ $(i,j \in \{1,2\})$ are randomly selected from
$\Z_p$. The Cloud Server also computes
\[
\vec \pi = (\pi_1,\pi_2) = ((1, g^{r_{11}}v^{-r_{21}}),(1,
g^{r_{12}}v^{-r_{22}})).
\]
and sends ($\vec c,\vec\pi$) as the response to the TPA.

TPA then verifies the response sent by the Cloud Server by checking
the equality of
\begin{equation}
\vec c \bullet \left(\begin{array}{cc}
1 & g\\
1 & v^{-1}
\end{array}
\right)=\iota_T(t_T) (\vec u \bullet \vec \pi) \label{verify}
\end{equation}
where $t_T$ represents the right hand side of Equation (\ref{eq2})
and $\iota_T$ denotes the following transformation:
\[
t_T \rightarrow \left(\begin{array}{cc}
1 & 1\\
1 & t_T
\end{array}
\right).
\]
The ``$\bullet$" operation is defined as follows: define a function
\[F((x_1,x_2),(y_1,y_2)) = \left(\begin{array}{cc}
e(x_1,y_1) &e(x_1,y_2)\\
e(x_2,y_1) & e(x_2,y_2)\end{array} \right) \] for $(x_1,x_2) \in
G_1^2$ and $(y_1,y_2) \in G_2^2$, and the ``$\bullet$" operation is
defined as
\[
\vec x \bullet \vec y = F(x_1,y_1)F(x_2,y_2).
\]

\smallskip\noindent\textbf{Correctness.}
To verify Equation (\ref{verify}),
\[
\mbox{Left} = \vec c \bullet \left(\begin{array}{cc}
1 & g\\
1 & v^{-1}
\end{array}
\right) = \left(
\begin{array}{cc}
e(c_{11},{ 1}) & e(c_{11},g) \\
e(c_{12},{ 1}) & e(c_{12},g) \\
\end{array}
\right)  \left(
\begin{array}{cc}
e(c_{21},{ 1}) & e(c_{21},v^{-1}) \\
e(c_{22},{ 1}) & e(c_{22},v^{-1}) \\
\end{array}
\right)
\]
\[ \mbox{Right} = \iota_T(t_T) F(u_1,\pi_1) F(u_2,\pi_2) ~~~~~~~~~~~~~~~~~~~~~~~~~~~~~~~~~~~~~~~~~~~~~~~
\]
\[=
\left(
\begin{array}{cc}
1 & 1 \\
1 & t_T \\
\end{array}
\right)
\left(
\begin{array}{cc}
1 & e(u,g^{r_{11}}v^{-r_{21}}) \\
1 & e(u^\alpha,g^{r_{11}}v^{-r_{21}}) \\
\end{array}
\right)
\left(
\begin{array}{cc}
1 & e(u^\tau,g^{r_{12}}v^{-r_{22}}) \\
1 & e(u^{\tau\alpha},g^{r_{12}}v^{-r_{22}}) \\
\end{array}
\right)~~~
\]
and we have
\[e(c_{11},{1})e(c_{21},{1})=1=1\cdot 1 \cdot 1\]
\[e(c_{12},{1})e(c_{22},{1})=1=1\cdot 1 \cdot 1\]
\begin{eqnarray*}
	e(c_{11},g)e(c_{21},v^{-1}) & = & e(u^{r_{11}+{r_{12}}\tau}, g)
	\cdot
	e(u^{r_{21}+{r_{22}}\tau},v^{-1}) \\
	& = &
	e(u,g^{r_{11}})e(u^{r_{21}},v^{-1})e(u^{\tau},g^{r_{12}})e(u^\tau,v^{-r_{22}})\\
	& = & e(u^{r_{11}+\tau r_{12}},g) e(u^{r_{21}+\tau r_{22}},v^{-1})
\end{eqnarray*}
\begin{eqnarray*}
	e(c_{12},g) e(c_{22},v^{-1}) & = & e(u^{\alpha(r_{11}+{r_{12}}\tau)}\sigma,g)e(u^{\alpha(r_{21}+{r_{22}}\tau)}u^{\mu'},v^{-1})\\
	& = & e(u^{\alpha(r_{11}+{r_{12}}\tau)},g)e(u^{\alpha(r_{21}+{r_{22}}\tau)},v^{-1}) e(\sigma,g)e(u^{\mu'},v^{-1})\\
	& = & t_T e(u^{\alpha r_{11}},g) e(u^{\alpha{r_{12}}\tau},g)
	e(u^{\alpha r_{21}}, v^{-1}) e(u^{\alpha{r_{22}}\tau},v^{-1})\\
	& = & t_T e(u^\alpha, g^{r_{11}}v^{-r_{21}})
	e(u^{\alpha\tau},g^{r_{12}}v^{-r_{22}})
\end{eqnarray*}

\subsection{DIC-Privacy of Our New Scheme}

Below we show that our new DIC protocol has the DIC-Privacy under
the symmetDIC external Diffie-Hellman (SXDH) assumption
\cite{GrothS08}. Let $gk = (\lambda, p, G_1, G_2,$ $ G_T, e, g_1, g_2)$
define a bilinear map $e: G_1\times G_2 \rightarrow G_T$ where $g_b$
is a generator of $G_b$ for $b = \{0,1\}$. The SXDH assumption holds
if for any polynomial time algorithm $\cal A$ and any $b \in \{1,2\}$ we
have
\[
|\Pr[x, y \leftarrow \Z_p^*: {\cal A}(gk, g_b^x, g_b^y, g_b^{xy}) = 1] -
\Pr[x, y, r \leftarrow \Z_p^*: {\cal A}(gk, g_b^x, g_b^y, g_b^r) = 1]| \le
\epsilon
\]
where $\epsilon$ is negligible in the security parameter $\lambda$.

\begin{theorem}
	Our new DIC protocol has DIC-Privacy if the SXDH problem is hard.
\end{theorem}

\begin{proof} Let $\cal A$ denote an adversary who has a non-negligible
	advantage $\epsilon$ in winning the Zero Knowledge Proof game, we construct another
	algorithm $\cal B$ which can solve the SXDH problem also with a
	non-negligible probability.
	
	$\cal B$ receives a challenge $gk, A = u^x, B = u^y, C = u^{z}$ where
	$gk = (p, G_1, G_2, G_T,$ $e,u, g)$ and $z$ is either $xy$ or a random
	element $\xi$ in $\Z_p$. $\cal B$ sets up the Zero Knowledge Proof game for $\cal A$ as
	follows
	\begin{enumerate}
		\item $\cal B$ uses the information in $gk$ to generate all the
		systems parameters and public/private keys as described in Wang et
		al.'s TPA scheme (Sec.~\ref{sec:wang}). \item $\cal B$ also sets the
		values of the commitment key $\vec u = (u_1,u_2)$ in our scheme as
		$u_1=(u,A)$ and $u_2=(B, C)$.
	\end{enumerate}
	
	Upon receiving the two files $F_0$ and $F_1$ from $\cal A$, $\cal B$
	simulates the game as follows. $\cal B$ generates a random file
	identifier $name$ and the file tag $t=name\|SSig_{ssk}(name)$, and
	uses $name$ and the secret key $x$ to compute the authenticators
	$\{\sigma^{(0)}_{i}\}$ (for $F_0$)  and $\{\sigma^{(1)}_{i}\}$ (for
	$F_1$) honestly. After that, $\cal B$ sends the file tag $t$ back to
	$\cal A$. Upon receiving the challenge $chal$ from $\cal A$, $\cal B$ computes
	$\mu_0'$, $\mu_1'$, and the corresponding aggregated authenticators
	$\sigma^{(0)}$ and $\sigma^{(1)}$ honestly. $\cal B$ then tosses a
	random coin $b \gets \{0,1\}$, and generates the response to $\cal A$ as
	follows.
	\begin{enumerate}
		\item Randomly choose $r_{11}, r_{12}, r_{21}, r_{22}$ from
		$\Z_p$.
		
		\item Compute $c_{11} = u^{r_{11}}B^{r_{12}}, c_{12} =
		A^{r_{11}}C^{r_{12}}\sigma^{(b)}, c_{21} = u^{r_{21}}B^{r_{22}},$ $
		c_{22} = A^{r_{21}}C^{r_{22}}u^{\mu_b'}$.
		
		\item Compute $\vec \pi = (\pi_1,\pi_2) = ((1,
		g^{r_{11}}v^{-r_{21}}),(1, g^{r_{12}}v^{-r_{22}}))$.
	\end{enumerate}
	
	$\cal B$ then sends the response $(\vec c, \vec \pi)$ to $\cal A$. If $\cal A$
	outputs $b'$ such that $b' = b$, then $\cal B$ outputs 1; otherwise
	$\cal B$ outputs $0$.
	
	\smallskip\noindent\textbf{Case 1: $z = xy$}. In this case, the distribution of the response $(\vec c, \vec \pi)$ is identically
	to that of a real response, and hence we have
	\[\Pr[b' = b] = 1/2 + \epsilon.\]
	\smallskip\noindent\textbf{Case 2: $z = \xi$}. In this case, the
	commitment scheme is perfectly hiding. That is, for a valid proof
	$(\vec c, \vec \pi)$ satisfying equation~\ref{verify}, it can be
	expressed as a proof for $(u^{\mu_0'}, \sigma_0)$ (with randomness
	$(r^0_{11}, r^0_{12}, r^0_{21}, r^0_{22})$), or a proof for
	$(u^{\mu_1'}, \sigma_1)$ (with randomness $(r^1_{11}, r^1_{12},
	r^1_{21}, r^1_{22})$). Therefore, we have
	\[\Pr[b' = b] = 1/2.\]
	Combining both cases, we have
	\begin{eqnarray*}
		& & \Pr[{\cal B}(gk, u^x, u^y, u^{xy}) = 1)] - \Pr[{\cal B}(gk, u^x, u^y,
		u^{\xi}) = 1)]\\
		& = & \Pr[b' = b| z = xy] - \Pr[b' = b| z = \xi]\\
		& = & \epsilon.
	\end{eqnarray*}
\end{proof}

\subsection{Soundness of the Protocol}

Having shown the Zero Knowledge Proof feature of the protocol, we have seen that
adversary $\cal A$ cannot distinguish the file that has been used by
the cloud server in an DIC proof. The remanning task is to prove
the ``soundness'' of the protocol. We say a protocol is sound if
it is infeasible for the cloud server to change a file without
being caught by the TPA in an auditing process. We formally define
the soundness games between a simulator $\cal B$ and an adversary $\cal A$
(i.e. the cloud server) as follows.

\begin{itemize}
	\item Key Generation. $\cal B$ generates a user key pair $(sk, pk)$ by running
	{\sf KeyGen}, and then provides $pk$ to $\cal A$.
	
	\item Phase 1. $\cal A$ can now interact with $\cal B$ and make at most
	$\ell$ Token Generation queries. In each query, $\cal A$ sends a file
	$F_i = \{m_{i1},m_{i2},\cdots,m_{in}\} (1 \le i \le \ell)$ to
	$\cal B$, which responds with the corresponding file tag $t_i$ and
	authentication tokens $\phi_i = \{\sigma_{ij}\}$ $(1 \le j \le
	n)$.

	\item Phase 2. $\cal A$ outputs a file $F^*$ and a file tag $t^*$ such
	that $t^* = t_i$ but $F^* \ne F_i$ for an $i \in [1,\ell]$ (i.e.
	at least one message block of $F_i$ has been modified by $\cal A$).
	$\cal B$ then plays the role as the verifier and executes the DIC
	protocol with $\cal A$ by sending a challenge $chal^* = \{j,\nu_j\}$
	which contains at least one index $j$ such that $F^*$ differs from
	$F_i$ in the $j$-th message block.
	
	\item Decision. Based on the proof ${\cal P}^*$ computed by $\cal A$,
	$\cal B$ makes a decision which is either True or False.
\end{itemize}

\begin{definition}
	We say a witness Zero Knowledge Proof DIC protocol is $\epsilon$-sound
	if
	\[\Pr[\cal B \mbox{ outputs True}] \leq \epsilon.\]
\end{definition}

Below we prove that our DIC protocol is sound under the co-CDH
assumption. Let $(p, G_1, G_2,$ $G_T, e, g_1, g)$ be the systems
parameters defined as above where $e: G_1 \times G_2 \to G_T$ is a
bilinear map. Let $\psi: G_2 \to G_1$ denote an efficiently
computable isomorphism such that $\psi(g) = g_1$.

\smallskip\noindent\textbf{Computational co-Diffie-Hellman (co-CDH) Problem}
on $(G_1,G_2)$: Given $g_1, u \in G_1$ and $g, g^a \in G_2$ as
input where $g_1$ and $g$ are generators of $G_1$ and $G_2$
respectively, $a$ is randomly chosen from $\Z_p$, and $u$ is
randomly chosen from $G_1$, compute $u^a \in G_1$.

\begin{theorem}
	The proposed witness Zero Knowledge Proof DIC protocol is
	$negl(\lambda)$-sound, where $negl(\lambda)$ is a negligible
	function of the security parameter $\lambda$, if the co-CDH
	problem is hard.
\end{theorem}

\begin{proof}
	Our proof is by contradiction. We show that if there exists an
	adversary $\cal A$ that can win the soundness game with a
	non-negligible probability, then we can construct another
	adversary $\cal B$ which can solve the co-CDH problem also with a
	non-negligible probability.
	
	
	According to the soundness game, $F^* =
	\{m_1^*,m_2^*,\cdots,m_n^*\}$ must be different from the original
	file $F_i = \{m_1,m_2,\cdots,m_n\}$ associated with $t^*$ (or
	$t_i$). That means there must exist an $i \in [1,n]$ such that
	$m_i^* \ne m_i$. Below we show that if $\cal A$ can pass the
	verification for $\mu^*$ where $\mu^* = \sum_{i \in I} \nu_im^*_i$
	and at lease one of $\{m^*_i\}_{i \in I}$ is modified by $\cal A$,
	then $\cal B$ can solve the co-CDH problem.
	
	$\cal B$ is given an instance of the co-CDH problem $(g_1, u, g, g^x)$
	where $g_1$ and $g$ are generators of $G_1$ and $G_2$ respectively
	such that $\psi(g) = g_1$, and $u$ is a random element in $G_1$.
	$\cal B$'s goal is to compute $u^x \in G_1$. $\cal B$ honestly generates
	the signing key pair $(spk,ssk)$, $\alpha, \tau \in \Z_p$ and the
	commitments key $u_1=(u,u^\alpha), u_2=(u^\tau, u^{\tau\alpha})$
	according to the protocol specification. $\cal B$ also sets $g^x$ as
	value of $v$ in the user public key, but the value of $x$ is
	unknown to $\cal B$. $\cal B$ then simulates the game as follows.
	\medskip
	
	\noindent {\bf Phase 1:} $\cal B$ answers $\cal A$'s queries in Phase 1 as
	follows. To generate a file tag $t_i$ for a file $F_i$, $\cal B$ first
	chooses $name$ at random and generates the file tag $t_i =
	name\|SSig_{ssk}(name)$. For each block $m_j (1 \leq j \leq n)$ in
	$F_i$, $\cal B$ chooses at random $r_j \in_R \Z_p$ and programs the
	random oracle
	\[
	H(W_j) = g_1^{r_j}/u^{m_j}.
	\]
	$\cal B$ then computes
	\[
	\sigma_j = (H(W_j)u^{m_{j}})^x =  (g_1^{r_j})^x = (\psi(v))^{r_i}.
	\]
	It is easy to verify that $\sigma_j$ is a valid authenticator with
	regards to $m_j$.
	\medskip
	
	\noindent {\bf Phase 2:} Suppose $\cal A$ outputs a response ${\cal P}^*
	= (\vec c, \vec \pi)$ for $t^*, \{m^*_i\}_{i \in I}$ and challenges
	$\{\nu_i\}_{i \in I}$ where at least one $m_i^*$ has been modified
	by the adversary. Denote $\mu^* = \sum_{i \in I} \nu_im^*_i$.
	
	Let $\mu = \sum_{i \in I} \nu_im_i$ and $\sigma = \Pi_{i \in
		I}\sigma_i^{\nu_i}$ denote the original file and authenticator
	that satisfy
	\begin{eqnarray}
	e(\sigma, g)=e((\prod_{i \in I} H(W_i)^ {\nu_i}) u^{\mu},v).
	\end{eqnarray}
	$\cal B$ then uses the value of $\tau$, which is used to generate the
	commitment key $\vec u$, to obtain $\sigma^* =
	c_{12}/c_{11}^\alpha$ and $u^{\mu^*} = c_{22}/c_{21}^\alpha$ from
	the commitment $\vec c = (c_1,c_2)$. Since ${\cal P}^*$ can pass
	the verification, from Equation \ref{verify} we have
	\begin{eqnarray}
	e(\sigma^*, g)=e((\prod_{i \in I} H(W_i)^ {\nu_i}) u^{\mu^*},v).
	\end{eqnarray}
	From Equation 5 and Equation 6, we can obtain
	\[
	e(\sigma^*/\sigma, g)
	= e(u^{\mu^* - \mu} ,v).
	\]
	Since $\cal B$ chooses the challenges $\nu_i$ randomly, with
	overwhelming probability $1 - 1/p$, $\mu^* = \sum_{i \in I}
	\nu_im^*_i \ne \sum_{i \in I} \nu_im_i = \mu$, and hence $\cal B$ can
	obtain
	\[u^x =(\sigma^*/\sigma)^{\frac{1}{\mu^* - \mu}}.\]
	
\end{proof}

\section{Conclusion}\label{sec:concl}
In this paper, we studied a new desirable security notion called
DIC-Privacy for remote data integrity checking protocols for cloud storage. We showed that several well-known DIC protocols cannot provide this property, which could render the privacy of user data exposed in an auditing process. We then proposed a new DIC protocol which can provide DIC-Privacy. Our construction is based on an efficient Witness Zero Knowledge Proof of Knowledge system. In addition, we also proved the soundness of the newly proposed protocol, which means the cloud server cannot modify the user data without being caught by the third party auditor in an auditing process.

\bibliographystyle{my}
\bibliography{references}

\begin{thebibliography}{10}

\bibitem{Arrington06}
M.~Arrington, `Gmail disaster: Reports of mass email deletions', {\em
  http://www.techcrunch.com/2006/12/28/gmail-disasterreports-of-mass-email-deletions/},
  (2006).

\bibitem{AtenieseBCHKKPS11}
Giuseppe Ateniese, Randal~C. Burns, Reza Curtmola, Joseph Herring, Osama Khan,
  Lea Kissner, Zachary N.~J. Peterson, and Dawn Song, `Remote data checking
  using provable data possession', {\em ACM Trans. Inf. Syst. Secur.}, {\bf
  14},  1--34, (2011).

\bibitem{BellareG92}
Mihir Bellare and Oded Goldreich, `On defining proofs of knowledge', in {\em
  Advances in Cryptology, Proc.\ CRYPTO 92, \emph{LNCS 740}}, pp. 390--420,
  (1992).

\bibitem{BonehLS04}
Dan Boneh, Ben Lynn, and Hovav Shacham, `Short signatures from the weil
  pairing', {\em J. Cryptology}, {\bf 17}(4),  297--319, (2004).

\bibitem{DeswarteQS04}
Yves Deswarte, Jean~Jacques Quisquater, and Ayda Saidane, `Remote integrity
  checking', in {\em Integrity and Internal Control in Information Systems VI},
  eds., Sushil Jajodia and Leon Strous, volume 140 of {\em IFIP International
  Federation for Information Processing},  1--11, Springer Boston, (2004).

\bibitem{FilhoB06}
D{\'e}cio Luiz~Gazzoni Filho and Paulo S{\'e}rgio Licciardi~Messeder Barreto,
  `Demonstrating data possession and uncheatable data transfer', {\em IACR
  Cryptology ePrint Archive},  150--159, (2006).

\bibitem{GrothS08}
Jens Groth and Amit Sahai, `Efficient non-interactive proof systems for
  bilinear groups', in {\em Advances in Cryptology, Proc.\ EUROCRYPT 2008,
  \emph{LNCS 4965}}, pp. 415--432, (2008).

\bibitem{JuelsK07}
Ari Juels and Burton S.~Kaliski Jr., `Pors: proofs of retrievability for large
  files', in {\em ACM Conference on Computer and Communications Security}, pp.
  584--597, (2007).

\bibitem{CSA}
\mbox{Cloud Security Alliance}, `Top threats to cloud computing', (2010).
\newblock http://www.cloudsecurityalliance.org.

\bibitem{Merkle80}
Ralph~C. Merkle, `Protocols for public key cryptosystems', in {\em IEEE
  Symposium on Security and Privacy}, pp. 122--134, (1980).

\bibitem{ShachamW08}
Hovav Shacham and Brent Waters, `Compact proofs of retrievability', in {\em
  Advances in Cryptology - ASIACRYPT}, pp. 90--107, (2008).

\bibitem{ShahBMS07}
M.~A. Shah, M.~Baker, J.~C. Mogul, and R.~Swaminathan, `Auditing to keep online
  storage services honest', in {\em Proc. of HotOSÕ07}, pp. 1--6, (2007).

\bibitem{WangChowWangRenLou}
Cong Wang, Sherman~S.M. Chow, Qian Wang, Kui Ren, and Wenjing Lou,
  `Privacy-preserving public auditing for secure cloud storage', {\em IEEE
  Transactions on Computers}.
\newblock Accepted for publication, doi: 10.1109/TC.2011.245.

\bibitem{WangRLL10}
Cong Wang, Kui Ren, Wenjing Lou, and Jin Li, `Toward publicly auditable secure
  cloud data storage services', {\em IEEE Network}, {\bf 24}(4),  19--24,
  (2010).

\bibitem{WangWRL10}
Cong Wang, Qian Wang, Kui Ren, and Wenjing Lou, `Privacy-preserving public
  auditing for data storage security in cloud computing', in {\em IEEE
  INFOCOM}, pp. 525--533, (2010).

\bibitem{WangWLRL09}
Qian Wang, Cong Wang, Jin Li, Kui Ren, and Wenjing Lou, `Enabling public
  verifiability and data dynamics for storage security in cloud computing', in
  {\em ESORICS}, pp. 355--370, (2009).

\bibitem{WangWRLL11}
Qian Wang, Cong Wang, Kui Ren, Wenjing Lou, and Jin Li, `Enabling public
  audibility and data dynamics for storage security in cloud computing', {\em
  IEEE Trans. Parallel Distrib. Syst.}, {\bf 22}(5),  847--859, (2011).

\end{thebibliography}
\end{document}